\input ./style/arxiv-general.cfg
\documentclass[aos,MSNbibl,nameyear,seceqn,dvips]{arximspdf}
\makeatletter
   \@ifpackageloaded{graphicx}{}{\usepackage{graphicx}}
\makeatother
\usepackage{algorithm}
\usepackage{algorithmic}

\doi{10.1214/15-AOS1358}
\volume{44}
\issue{1}
\pubyear{2016}
\firstpage{87}
\lastpage{112}
\docsubty{FLA}

\makeatletter
\newcommand{\eqref}[1]{(\ref{#1})}
\newproclaim{definition}{Definition}
\newtheorem{theorem}{Theorem}
\newtheorem{lemma}{Lemma}
\newtheorem{proposition}{Proposition}
\newtheorem{corollary}{Corollary}
\newproclaim{remark}{Remark}
\newproclaim{remarks}{Remarks}

\newproclaim{example}{Example}

\makeatother

\begin{document}
\begin{frontmatter}

\title{Inference using noisy degrees: Differentially private $\beta$-model and synthetic graphs}
\runtitle{Private synthetic graphs}

\begin{aug}
\author[A]{\fnms{Vishesh} \snm{Karwa}\corref{}\thanksref{T1,T2}\ead[label=e1]{vishesh@cmu.edu}}
\and
\author[B]{\fnms{Aleksandra} \snm{Slavkovi\'{c}}\thanksref{T1}\ead[label=e2]{sesa@psu.edu}}
\runauthor{V. Karwa and A. Slavkovi\'{c}}
\thankstext{T1}{Supported in part by NSF Grant BCS-0941553 to the
Department of Statistics,
Pennsylvania State University.}
\thankstext{T2}{Supported in part by the Singapore National Research
Foundation under
its International Research Centre \@ Singapore Funding Initiative and
administered by the
IDM Programme Office through a grant for the joint Carnegie
Mellon/Singapore Management
University Living Analytics Research Centre.}
\affiliation{Carnegie Mellon University and Pennsylvania State University}
\address[A]{Department of Statistics\\
Carnegie Mellon University\\
132G Baker Hall\\
Pittsburgh, Pennsylvania 15213\\
USA\\
\printead{e1}}

\address[B]{Department of Statistics\\
Pennsylvania State University\\
421A Thomas Bldg.\\
University Park, Pennsylvania 16802\\
USA\\
\printead{e2}}
\end{aug}

%
\received{\smonth{8} \syear{2014}}
%
\revised{\smonth{6} \syear{2015}}

%
\begin{abstract}
The $\beta$-model of random graphs is an exponential family model with the
degree sequence as a sufficient statistic. In this paper, we contribute
three key results. First, we characterize conditions that lead to a
quadratic time algorithm to check for the existence of MLE of the
$\beta
$-model, and show that the MLE never exists for the degree partition
$\beta$-model. Second, motivated by privacy problems
with network data, we derive a differentially private estimator of the
parameters of $\beta$-model, and show it is consistent and
asymptotically normally distributed---it achieves the same rate of
convergence as the nonprivate estimator.
We present an efficient algorithm for the private estimator that can be
used to release synthetic graphs. Our techniques can also be
used to release degree distributions and degree partitions accurately
and privately, and to
perform inference from noisy degrees arising from contexts other than
privacy. We
evaluate the proposed estimator on real graphs and compare it with a current
algorithm for releasing degree distributions and find that it does
significantly better.
Finally, our paper addresses shortcomings of current approaches to a
fundamental problem of how to perform valid statistical inference from
data released by privacy mechanisms, and lays a foundational groundwork
on how to achieve optimal and private statistical inference in a
principled manner by modeling the privacy mechanism; these principles
should be applicable to a class of models beyond the $\beta$-model.
\end{abstract}

%
%
\end{frontmatter}

\section{Introduction and motivation}
Random graph models whose sufficient statistics are degree sequences, $d$,
such as the $p_1$ model for directed graphs or its special case, the
$\beta$-model for undirected graphs [\citet
{hollandp1,chatterjee2011random,olhede2012degree,rinaldo2011maximum}]
are commonly used
in modeling of real world networks. Although there is evidence that $d$ alone
does not capture all the structural information in a graph [e.g.,
\citet
{snijders2003accounting}], in many cases it is the only information
available and every other structural property of a graph is estimated from
random graph models based on $d$. In more general cases, random graph models
based on $d$ serve as a natural starting point for modeling networks;
they may also serve as null models for hypothesis testing
[\citet{perry2012null, zhangandchen}]. However, the degrees may carry
confidential and sensitive information, and thus limit our ability to share
such data more widely for the purpose of statistical inference.
For
example, in epidemiological studies of sexually transmitted disease [e.g.,
see \citet{helleringer2007sexual}], a survey collects information on the
number of sexual partners of an individual, which provides an estimate
of the
degree of each node that is then used for modeling and reconstruction
of a sexual network. The benefits of analyzing such networks are clear [e.g.,
\citet{goodreau2009birds}], but releasing such sensitive information raises
significant privacy concerns [e.g., (\citet{narayanan2009anonymizing}].

Data privacy is a growing problem due to the large amount of data being
collected, stored, analyzed and shared across multiple domains.
Statistical Disclosure Control (SDC) aims at designing data sharing mechanisms
that address the trade-off between minimizing the risk of
disclosing sensitive information and maximizing of data utility; for more
details on SDC methodology,
see, for example, \citet{willenborg96, fienberg10,ramanayake10} and
\citet{hundepool2012sdc}.
More recently, data privacy research has evolved with a focus on
designing mechanisms that satisfy
some \emph{rigorous} notions of privacy but at the same time provide
meaningful utility.

Differential Privacy (DP)
[\citet{DMNS06}] has emerged as a key rigorous definition
of privacy and as a way to inform the design of privacy
mechanisms with pre-specified worst case disclosure risk. However,
existing DP mechanisms are designed with a focus on estimating accurate
summary statistics of the data, as opposed to estimating parameters of
a model that are consistent and have correct confidence intervals; see
\citet{smith2008efficient} and \citet{vu2009differential} for exceptions.
As recently shown by \citet{duchi2013local}, estimating parameters of
models (that correspond to population quantities) and estimating
summary statistics are fundamentally different problems, especially in
the privacy context.
However, the privacy mechanism is typically ignored and
the perturbed statistics are used for subsequent analyses.
Among many potential problems, ignoring the privacy mechanism can lead
to invalid even nonexistent parameter estimates, as initially
demonstrated in \citet{FienbergRinYang}, \citet{karwapsd2012} and in
this paper.

This paper addresses the above mentioned fundamental problem of
performing valid statistical inference using data released by a
differentially private mechanism.
Our work demonstrates that to obtain optimal parameter estimates by
using data shared by privacy preserving mechanisms, new estimation
procedures must be derived for specific classes of inference problems
by modeling the privacy mechanism as a nonlinear measurement error
process. The nonlinearity arises from the fact that noise is usually
added to the sufficient statistics, as opposed to the data [see also
\citet{carroll2006measurement}].
We illustrate the proposed principles in the context of special but
important case of sharing network data using differential privacy;
however, these principles are applicable beyond the specific privacy
mechanism and the models considered here.

For network data, DP comes in two
variants: Edge Differential Privacy [e.g., see \citet{nissim2007smooth}]
and Node Differential Privacy [e.g., see
\citet{kasiviswanathan2013analyzing}], designed to limit disclosure of edge
and node (along with its edges) information, respectively, in a graph $G$.
We focus on the edge differential privacy with a goal of estimating
the parameters of the $\beta$-model of random graphs whose sufficient
statistics are network's degrees $d$. One of the popular ways of
releasing $d$ (and in general any summary statistic) to protect privacy
is to release $z = d+e$, where $e$ is some noise. In some cases, $z$ is
post-processed to reduce error [e.g., see \citet{hay2009} for release of
degree partitions] with the end goal to obtain an approximate estimate
of the summary statistic of the data. However, the end goal of a
statistical inference is not the estimation of statistics, in fact, the
sufficient statistics are the starting point. Without any additional
tools, the analyst is forced to directly use the noisy summary
statistic $z$ for inference.
We present techniques to take into account the noise addition process
and thereby consistently compute the maximum likelihood estimates (MLE)
of the $\beta$-model from a noisy degree sequence.
The following are the more specific contributions of this paper:
\begin{longlist}[1.]
\item[1.] In Theorem~\ref{thm:beta.mle.exist}, we derive necessary and
sufficient conditions for the
existence of MLE of the $\beta$-model, a result applicable beyond the
privacy context. These conditions are computationally more efficient
than those of \citet{rinaldo2011maximum}, which are more general, but
computationally intractable. This result gives insights into the
conditions when the parameter estimates do not exist due to noisy
statistics arising from privacy or possibly from sampling and
censoring [\citet{handcock2010modeling}].

\item[2.] Using the result on existence of MLE, we illustrate that ignoring
the privacy mechanism and directly using the noisy statistic $z$ for
inference may lead to issues such as nonexistence of MLE of the $\beta
$-model. We also illustrate that the customary practice
of simply minimizing the $L_1$ and/or $L_2$ distance between original
and noisy statistics are not sufficient measures to guarantee
statistical utility, and thus a valid inference. In particular, to
obtain optimal and valid parameter estimates, the privacy mechanism
must be explicitly taken into account when estimating the sufficient
statistics from their noisy versions.

\item[3.] By modeling the privacy mechanism as a (known) measurement error
process, we obtain a private maximum likelihood estimate $\hat{d}$
of the degree sequence~$d$, from its noisy counterpart $z$. In
Theorem~\ref{thm:opt} and Algorithm \ref{alg:main}, we show that this
estimation problem can be solved efficiently, using a well-known
characterization of degree sequences due to \citet{havel} and \citet
{hakimi}. This is a nonstandard maximum likelihood estimation problem
where the parameter set is discrete and its dimensionality increases
with the sample size. Using simulation studies, we show that $\hat{d}$
has smaller error and greater statistical utility when compared to
using $z$ directly for parameter estimation.

\item[4.] In Theorems \ref{thm:consistent} and \ref{thm:clt}, we derive a
differentially private consistent and asymptotically normal estimator
$\hat{\beta}_{\varepsilon}$ of the parameters of the $\beta$-model of
random graphs, by using the proposed estimated $\hat{d}$ (instead of
$z$). $\hat{\beta}_{\varepsilon}$ then can be used to generate valid
synthetic graphs. Consistency of the usual MLE of $\beta$, without any
privacy constraints, 
was shown by \citet{chatterjee2011random} and its asymptotic normality
was established in \citet{Yan01062013}. Critically, since the proposed
$\hat{\beta}_{\varepsilon}$ achieves the same rate as the nonprivate
estimator, we show that asymptotically privacy comes at no additional
cost in this setting.
\end{longlist}

The rest of the paper is organized as follows. In Section~\ref
{prelim}, we
introduce the notation and the key results on the existence of MLE of the
$\beta$-model and inference from noisy statistics. In Section~\ref
{edp}, we
describe our privacy model. 
Section~\ref{main} forms the core of the paper where we present our main
results on estimating differentially private parameters of the $\beta$-model and on generating synthetic graphs. In
Section~\ref{compareWithHay}, we extend our algorithm to release degree
partitions and compare it to that of \citet{hay2009}. In Section~\ref
{experiments}, we evaluate our proposed estimators on real graphs.
In Section~\ref{conclusion}, we briefly
discuss avenues for future work, including the challenges in extending
our work to larger class of $\beta$-models. Proofs are presented in
Section~\ref{proofs} and the supplementary material [\citet{karwaslavsupp}].

\section{Statistical inference with degree sequences}
\label{prelim}

Let $G_n$ denote a \textit{simple}, \textit{labeled} undirected graph
on $n$ nodes
and let $m$ be the number of edges in the graph. Let $V$ be the vertex set
and $E$ be the edge set of the graph. A \textit{simple} graph is a graph
with no
self-loops and multiple edges, that is, for any $i \in[n]$, $(i,i)
\notin E$,
and $|\{(i,j): (i,j) \in E\}| = 1$. A \textit{labeled} graph is a graph
with a
fixed ordering on its nodes, that is, there is a fixed mapping from $V$
to $\{1,
\ldots, n\}$. All the graphs considered in this paper are simple and
undirected. Let $\mathcal{G}$ denote the set of all such graphs. The distance
between two graphs $G$ and $G'$ is defined as the number of edges on which
the graphs differ and is denoted by $\delta(G,G')$. $G$ and $G'$ are
said to
be neighbors of each other if the distance between them is at most $1$.
The degree $d_i$ of a node $i$ is the number of nodes connected to it.

%
\begin{definition}[(Degree sequence and degree partition)] 
Consider a labeled graph with label $\{1, \ldots,n\}$. The degree
sequence of a graph $d$ is defined as the sequence of degrees of each
node, that is, $d = \{d_1, \ldots, d_n \}$. The degree sequence ordered
in nonincreasing order is called the degree partition and is denoted
by $\bar{d}$, that is, $\bar{d} = \{d_{(1)}, \ldots, d_{(n)} \}$ where
$d_{(i)}$ is the $i$th largest degree. 
\end{definition}

Given a degree sequence $d$, there can be more than one graph with different
edge-sets $E$, but the same degree sequence~$d$. Each such
graph is called a realization of $d$. Let $\mathcal{G}(d)$ be
the set of simple graphs on $n$ vertices with degree sequence~$d$. Not
every integer sequence of length $n$ is a degree sequence. Sequences
that can
be realized by a simple graph are called \textit{graphical degree
sequences}. Graphical degree sequences have been studied in depth and admit
many characterizations. One of the characterizations called the
Havel--Hakimi criteria, due to \citet{havel} and \citet{hakimi}, is
central to
the proof of Algorithm \ref{alg:main} that estimates a graphical degree
sequence from the noisy sequence $z$; see the proof of Theorem~\ref
{alg:main} in the supplementary material [\citet{karwaslavsupp}] for the
statement of the characterization. We denote the set of all graphical
degree sequences of size $n$ by $\mathit{DS}_n$ and the set of all graphical
degree partitions of size $n$ by $\mathit{DP}_n$.

\subsection{Statistical inference with the \texorpdfstring{$\beta$}{beta}-model}
\label{beta.model}
One of the simplest random graph models involving the degree sequence
is called the $\beta$-model, a term coined by \citet
{chatterjee2011random}. We can describe this model in terms of
independent Bernoulli random variables. Let $\beta= \{\beta_1, \ldots,
\beta_n\}$ be a fixed point in~$\mathbb{R}^n$. For a random graph on
$n$ vertices, let each edge between nodes $i$ and $j$ occur
independently of other edges with probability
\[
p_{ij} = \frac{e^{\beta_i + \beta_j}}{1+ e^{\beta_i + \beta_j}},
\]
where $\{\beta_1, \ldots, \beta_n\}$ is the vector of parameters.

This model admits many different characterizations. For example, it
arises as a special case of $p_1$ models [\citet{hollandp1}] and a
log-linear model [\citet{rinaldo2011maximum}]. It is also a special case
of the discrete exponential family of distributions on the space of
graphs when the degree sequence is a sufficient statistic. Thus, if $G$
is a graph with degree sequence $\{d_1, \ldots, d_n\}$, then the
$\beta
$-model is described by
\[
P(G = g) \propto\exp{\sum_{i=1}^n{d_i
\beta_i}}.
\]
We can also consider a version of the $\beta$-model where the degree
partition $\bar{d}$ is a sufficient statistic. Such a model may be used
if the ordering
of the nodes is irrelevant.

In modeling real world networks, there are two very common inference tasks
associated with the $\beta$-model:
\begin{longlist}[1.]
\item[1.] Sample graphs from $\mathcal{U}(d)$---the uniform distribution
over the set of all graphs with degree sequence $d$.
\item[2.] Estimate parameters of the $\beta$-model using $d$ and generate
synthetic graphs from the $\beta$-model.
\end{longlist}

These tasks are useful, for example, in modeling network when the
degree sequence is the only available information [\citet
{helleringer2007sexual}], and in performing goodness-of-fit testing
of more general network models [Hunter, Goodreau and
Handcock (\citeyear{huntergof})].
A natural question to ask is under what conditions on $d$ and $\bar{d}$
are these two tasks possible: (a) Under what conditions does the MLE of
the $\beta$-model exist? and (b) When is it possible to sample from
$\mathcal{U}(d)$? In the next section, we study the conditions on $d$
and $\bar{d}$ that allow us to perform these inference tasks.

\subsection{Existence of MLE of the \texorpdfstring{$\beta$}{beta}-model}
\label{existenceOfMLE}
Let $\hat{\beta}(d)$ denote the maximum likelihood estimate of $\beta$
obtained using $d$. If we consider the degree partition version of the
$\beta$-model, the MLE is denoted by $\hat{\beta}(\bar{d})$. From the
properties of exponential families, it follows that $\hat{\beta}(d)$
must satisfy the following moment equations:
%
%
\begin{equation}
\label{mle.eq} d_i = \sum_{j \neq i}
\frac{e^{\hat{\beta}_i + \hat{\beta_j}}}{1+
e^{\hat{\beta_i} + \hat{\beta_j}}}.
\end{equation}
A solution to these equations can be obtained in many ways. Most of
them require iterative procedures [\citet
{HunterBradley,chatterjee2011random}]. These procedures do not converge,
or may converge to a meaningless value, when the MLE does not exist.

In Theorem~\ref{thm:beta.mle.exist}, we describe necessary and
sufficient conditions for existence of the MLE of the $\beta$-model.
These conditions lead to an $O(n^2)$ algorithm to check for the
existence of the MLE for the degree sequence $\beta$-model and show
that the MLE never exists for the degree partition $\beta$-model. To
the best of our knowledge, this is the first efficient algorithm for
checking the existence of MLE of the $\beta$-model. The proof of
Theorem~\ref{thm:beta.mle.exist} is in Section~\ref{proof:beta.mle.exist}.

From the theory of exponential families
[\citet{nielsen1978information}], it follows that $\hat{\beta}(d)$ exists
if and only if
$d$ lies in the relative interior of convex hull of
$\mathit{DS}_n$. Although the facets of $\operatorname{Conv}(\mathit{DS}_n)$ are completely
characterized in \citet{mahadev1995threshold}, one cannot use the linear
inequality description of $\operatorname{Conv}(\mathit{DS}_n)$ to check if $d$
lies in the
relative interior. This is because $\operatorname{Conv}(\mathit{DS}_n)$ is a complex
combinatorial object and the number of facet defining inequalities
[given in equation (\ref{eq:desDS})] are at least exponential in $n$.
\citet{rinaldo2011maximum} use results from the existence of MLE of
discrete exponential families [\citet{rinaldo2009geometry}]
to devise an algorithm to check for the existence of MLE in what they
refer to as a \textit{generalized}
$\beta$-model. Their algorithm is based on the so-called ``Cayley
embedding'' which is a reparametrization of the $\beta$-model as a
log-linear model. Although general, their algorithm works only for
graphs up to a few hundreds of nodes, and its computational complexity
is unknown.

The key technique that we use for proving Theorem~\ref
{thm:beta.mle.exist} is to study an ``asymmetric'' part of
$\operatorname{Conv}(\mathit{DS}_n)$. Specifically, we work with $\operatorname
{Conv}(\mathit{DP}_n)$, the convex hull
of degree partitions, instead of $\operatorname{Conv}(\mathit{DS}_n)$. Intuitively,
$\operatorname{Conv}(\mathit{DP}_n)$ can be considered as a ``asymmetrized''
version of
$\operatorname{Conv}(\mathit{DS}_n)$---every permutation equivalent degree
sequence is mapped
to a single degree partition [see also \citet
{bhattacharya2006polytope}]. This asymmerization, remarkably, allows us
to characterize the boundary of $\operatorname{Conv}(\mathit{DS}_n)$, and at
the same time,
greatly reduce the computational complexity. We conjecture that this
technique of asymmerizing a polytope can be extended to other discrete
exponential families to derive efficient algorithms that characterize
their boundary.

%
\begin{theorem}
\label{thm:beta.mle.exist}
Let $G$ be a graph. Let $d$ be its degree sequence and $\bar{d}$ be the
corresponding degree partition obtained by ordering the terms of $d$ in
a nonincreasing order. Consider the following set of inequalities:
%
%
\begin{eqnarray}
\label{mle.ineq}&& \bar{d}_i > 0\quad \mbox{and}\quad
\bar{d}_i < n-1 \qquad\forall i\quad \mbox{and}
\nonumber
\\[-8pt]
\\[-8pt]
\nonumber
&&\sum_{i=1}^k{\bar{d}_i} -
\sum_{i=n-l+1}^{n}{\bar{d}_i} <
k(n-1-l) \qquad\mbox{for } 1 \leq k+l \leq n,
\end{eqnarray}
%
%
\begin{equation}
\label{order.ineq} \bar{d}_{i+1} -
\bar{d}_{i} < 0\qquad \mbox{for } i = 1 \mbox{ to } n.
\end{equation}

The following statements are true:
\begin{longlist}[1.]
\item[1.] The MLE of the degree partition $\beta$-model $\hat{\beta
}(\bar
{d})$ exists iff $\bar{d}$ satisfies the system of inequalities in
(\ref
{mle.ineq}) and (\ref{order.ineq}). In particular, the MLE for the degree
partition $\beta$-model never exists.
\item[2.] If the MLE of the degree sequence $\beta$-model $\hat{\beta}(d)$
exists, then $\bar{d}$ satisfies the system (\ref{mle.ineq}).
\item[3.] If $\bar{d}$ satisfies the system (\ref{mle.ineq}), then $\hat
{\beta
}(d)$ exists for any $d = \pi\bar{d}$ where $\pi$ is a permutation on
$\{1, \ldots, n\}$.
\end{longlist}
\end{theorem}

\begin{remarks*}
\begin{longlist}[1.]
\item[1.] The system of inequalities in equation (\ref{mle.ineq})
are central to the results of Theorem~\ref{thm:beta.mle.exist}. There are
only $O(n^2)$ inequalities to check, as opposed to exponentially many
inequalities that describe $\operatorname{Conv}(\mathit{DS}_n)$. Thus, an
important practical
consequence of this result is the first quadratic time algorithm to
detect the boundary points of $\operatorname{Conv}(\mathit{DS}_n)$ and check
for the existence
of MLE of the degree sequence $\beta$-model.
\item[2.] Statement 3, the converse condition in Theorem~\ref
{thm:beta.mle.exist} is stronger than statement 2. It implies that if
$\bar{d}$ satisfies the system (\ref{mle.ineq}), then the MLE of $\beta$
computed using any permutation of $\bar{d}$ exists.

\item[3.] Theorem~\ref{thm:beta.mle.exist} does not imply that $d$ is in
$\mathit{ri}(\operatorname{Conv}(\mathit{DS}_n))$ if and only if $\bar{d}$ is in
$\mathit{ri}(\operatorname{Conv}(\mathit{DP}_n))$. In
fact, this is not true---no (graphical) degree partitions exists in
the relative interior of $\operatorname{Conv}(\mathit{DP}_n)$; all degree
partitions lie on at
least one of the boundaries defined by equation (\ref{order.ineq}).

\item[4.] When we observe a single graph, the MLE for the degree partition
$\beta$-model never exists. From this point onward, we will use the
term ``MLE of $\beta$'' to mean the MLE of the degree sequence $\beta$-model, even when using a degree partition, since every degree partition
is also a degree sequence.
\item[5.] The degree distribution is the histogram of degree partition, and
furthermore the degree distribution and the degree partition are one to
one transformations of each other, one can be obtained from the other
via a nonlinear transformation. Most recently, \citet
{sadeghi2014statistical} show that the MLE of the degree distribution
model also never exists which complements our results on the degree partition.
\end{longlist}
\end{remarks*}

\subsection{Sampling from $\mathcal{U}(d)$}

Sampling graphs from the set $\mathcal{U}(d)$ is possible only if the set
$\mathcal{G}(d)$ is nondegenerate. Moreover, for there to exist a
nontrivial probability distribution on this set, its cardinality
should be
greater than 1. Proposition~\ref{prop:suffcondSampling} presents
sufficient conditions on $d$ 
under which this is true; the proof appears in Section~II of the supplementary material [\citet{karwaslavsupp}].

%
\begin{proposition}
\label{prop:suffcondSampling}
Let $d$ be a sequence of real numbers. 
Consider the set $\mathcal{G}(d)$, the set of all simple graphs with
degree sequence equal to $d$. If $d$ is a point in $\mathit{DS}_n$, and if $d$
lies in the relative interior of $\operatorname{Conv}(\mathit{DS}_n)$, then
$|\mathcal{G}(d)| >
1$. 
\end{proposition}

\subsection{Inference using noisy statistics}
Theorem~\ref{thm:beta.mle.exist} and Proposition~\ref{prop:suffcondSampling}
give sufficient conditions for estimating parameters of the $\beta
$-model and
for sampling from the space of related graphs. However, in many real world
applications, the exact degree sequence $d$ of a graph is not
available. Instead, we observe a ``noisy'' sequence $z$ either due to sampling
issues or due to privacy constraints. 
Corollary~\ref{cor:sufficientconditions} gives sufficient conditions
for obtaining
valid inference in the $\beta$-model when using such ``noisy'' sequences.
%

\begin{corollary}
\label{cor:sufficientconditions}
Let $z$ be any sequence of integers of length $n$. Consider the
following two inference task: (1) Estimating the MLE of $\beta$-model
using $z$. (2) Sampling from the set $\mathcal{U}(z)$. A sufficient
condition to ensure that the MLE exists and $\mathcal{U}(z)$ is
nonempty is that $z$ is a point in $\mathit{DS}_n$ and lies in the relative
interior of convex hull of $\mathit{DS}_n$.
\end{corollary}

In Section~\ref{main}, we consider the case where $z$ is a noisy degree
sequence obtained by applying a differentially
private mechanism to $d$. We discuss in more detail why directly
using $z$ instead of $d$ typically leads to invalid inference and apply
the results of this section to obtain valid statistical inference by
finding an estimate of $d$ that satisfies conditions of Corollary~\ref
{cor:sufficientconditions}.

\section{Edge differential privacy}
\label{edp}
Differential privacy has become one of the most popular models of
reasoning formally about privacy. In a typical interactive setting,
data users can ask {\textit{queries}} about the data, which can be in the
form of sufficient statistics, and they would receive back
differentially private answers. This type of a privacy mechanism can be
formalized as a family of conditional probability distributions, which
define a distribution on the answers, conditional on the data; for a
statistical overview of differential privacy; see \citet{wassermanzhou}.

In this paper, we focus on edge differential privacy (EDP) where the
goal is to protect the topological
information of the graph. 
EDP is defined to limit disclosure related to presence or absence of
edges in a graph (or relationships between nodes) as the following
definition illustrates.

%
\begin{definition}[(Edge differential privacy)] Let $\varepsilon> 0$. A
randomized mechanism (or a family of conditional probability
distributions) $\mathcal{Q}(\cdot|G)$ is $\varepsilon$-edge differentially
private if
\[
\mathop{\operatorname{sup}}_{G, G' \in\mathcal{G}, \delta(G,G') = 1}
\mathop{\operatorname{sup}}_{S\in
\mathcal{S}} \log
\frac{\mathcal{Q}(S|G)} {\mathcal{Q}(S|G')} \leq \varepsilon,
\]
where $\mathcal{S}$ is the set of all possible outputs (or the range of
$\mathcal{Q}$).
\end{definition}

$\varepsilon$ is the privacy parameter that, as we see below, controls the
amount of noise added to a statistic; small value
of $\varepsilon$ means more privacy protection, but leads to larger noise
in the statistic being released. Roughly, EDP requires that any output
of the mechanism
$\mathcal{Q}$ on two neighboring graphs should be close to each other. Along
the lines of Theorem~2.4 in \citet{wassermanzhou}, one can show that
EDP makes
it nearly impossible to test the presence or absence of an edge in the graph,
thus providing protection.

The most common mechanism to release the output of any statistic $f$ under
differential privacy is the Laplace mechanism [e.g., see \citet{DMNS06}]
which adds
continuous Laplace noise proportional to the \textit{global sensitivity}
of $f$.

%
\begin{definition}[(Global sensitivity)]
Let $f: \mathcal{G} \rightarrow\mathbb{Z}^k$. The global sensitivity
of $f$ is defined as
\[
GS(f) = \max_{\delta(G,G')=1}{\bigl\Vert f(G)-f \bigl(G' \bigr)
\bigr\Vert_1},
\]
where $\Vert\cdot\Vert_1$ is the $L_1$ norm.
\end{definition}

Here, we propose to use a variant of this mechanism to achieve EDP by
adding discrete Laplace
noise, as described in Lemma~\ref{thm:lap.mech}, to the degree sequence
of a
graph (see Algorithm \ref{noisyInference} in Section~\ref{degree.mle}). \citet
{ghosh2009universally} analyzed the discrete Laplace mechanism for one-
dimensional counting queries and showed that it is universally optimal
for a large class of utility metrics. 
The proof of Lemma~\ref{thm:lap.mech} is given in Section~I of the supplementary material [\citet{karwaslavsupp}].

%
\begin{lemma}[(Discrete Laplace mechanism)]
\label{thm:lap.mech}
Let $f: \mathcal{G} \rightarrow\mathbb{Z}^k$. Let $Z_1, \ldots, Z_k$
be independent and identically distributed discrete Laplace random
variables with p.m.f. defined as follows:
\[
P(Z = z) = \frac{1 - \alpha}{1+\alpha}\alpha^{|z|},\qquad z \in\mathbb {Z}, \alpha
\in(0,1).
\]
Then the algorithm which on input $G$ outputs $f(G)+ (Z_1, \ldots,
Z_k)$ is $\varepsilon$-edge differentially private, where $\varepsilon=
-GS(f) \log\alpha$.
\end{lemma}

One nice property of differential privacy is that any function of a
differentially private mechanism is also differentially private.

%
\begin{lemma}[{[\citet{DKMMN06,wassermanzhou}]}]\label{thm:post-processing}
Let $f$ be an output of an $\varepsilon$-differentially private mechanism
and $g$ be any function. Then $g(f(G))$ is also $\varepsilon
$-differentially private.
\end{lemma}

By using Lemma~\ref{thm:post-processing}, we can ensure that any
post-processing done on the noisy degree sequences obtained as an
output of a
differentially private mechanism is also differentially private. In
particular, this means that applying the proposed Algorithm \ref{alg:main}
to the output of a differentially private mechanism also preserves
differential privacy.

\section{Estimating parameters of the \texorpdfstring{$\beta$}{beta}-model using noisy degree
sequences and releasing synthetic graphs}
\label{main}

In this section, we present our main results on obtaining consistent
and asymptotically normal differentially-private\break MLEs for the $\beta$-model.
These results support two main objectives: (1) To achieve statistical
inference that is both optimal and private for the $\beta$-model, and
(2) to release synthetic graphs from the $\beta$-model in a
differentially private manner.

Our approach is based on three steps. In the first step, we release the
degree sequence, which is a
sufficient statistic of the $\beta$-model, using the discrete Laplace
mechanism described in Lemma~\ref{thm:lap.mech}. In the second step, we
model the Laplace mechanism as a measurement error on the sufficient
statistics and ``de-noise'' the noisy sufficient statistic by using
maximum likelihood estimation. In the third step, the de-noised
sufficient statistic is used to estimate the parameters of the $\beta
$-model from which synthetic graphs can be generated. Since each of
these steps uses only the output of a differentially private algorithm,
by Lemma~\ref{thm:post-processing}, the generated synthetic graphs are
also differentially private. Step 2 of modeling the privacy mechanism
as a measurement error process and re-estimating the degree sequence is
critical, as we show in the proofs of Theorems \ref{thm:consistent} and
\ref{thm:clt}, since it allows the third step to produce consistent and
asymptotically normal parameter estimates.
In the next subsections, we look at each of these steps in detail and
describe the associated algorithms and theoretical results.

\subsection{Releasing the degree sequence privately}
\label{degree.mle}
Since the degree sequence $d$ (or degree partition $\bar{d}$) is a
sufficient statistic of the $\beta$-model, the first step releases
these statistics under differential privacy via Algorithm \ref
{alg:lap.degree}. We use the discrete Laplace mechanism (Lemma~\ref
{thm:lap.mech}). 
The global sensitivity of both $d$ and $\bar{d}$ is 2 since adding or
removing an edge can change the degree of at most two nodes, by 1 each.

\begin{algorithm}[t]
\caption{\textit{Input}: A graph $G$ and privacy parameter
$\varepsilon$.
\textit{Output}: Differentially private answer to the degree
sequence of $G$}\label{noisyInference}
\begin{algorithmic}[1]
\label{alg:lap.degree}
\STATE Let $d = \{d_1, \ldots, d_n\}$ be the degree sequence of $G$
\FOR{$i = 1 \to n$}
\STATE Simulate $e_i$ from discrete Laplace with $\alpha= \exp
(-\varepsilon/2)$
\STATE Let $z_i = d_i + e_i$
\ENDFOR
\RETURN$z = \{z_1, \ldots, z_n\}$
\end{algorithmic}
\end{algorithm}

Can we use $z$, a differentially private output of the degree sequence $d$
released by Algorithm \ref{alg:lap.degree}, directly for inference and
generate synthetic graphs? Most work on differential privacy advocates
using $z$ or some post-processed form of $z$ as a ``proxy'' of $d$ for
inference. 
This, however, ignores the noise addition process. Furthermore, a more
serious issue is that $z$ may not satisfy the conditions of
Corollary~\ref{cor:sufficientconditions}.

To understand how $z$ fails the conditions of Corollary~\ref
{cor:sufficientconditions},
consider task (1) from Section~\ref{prelim} where the goal is to
simulate random graphs from the $\mathcal
{U}(d)$ by using the output $z$ instead of $d$. Recall that $\mathcal
{U}(d)$ is nonempty if and only if $d$ is a point in $\mathit{DS}_n$, that is,
$d$ is a graphical sequence. What are the chances that $z$ is
graphical? If $z$ is a sequence of positive integers, the chances are
asymptotically at best $50 \%$; see \citet{arratia2005likely}. In the
present case, $z$ is supported on the set of integers, $\mathbb{Z}_n$
as it is obtained by adding discrete Laplace noise to $d$. Hence, it is
quite unlikely for $z$ to even be in $\operatorname{Conv}(\mathit{DS}_n)$.
Thus, in many cases
$z$ cannot be used directly to perform task (1). 

How about task (2) of estimating $\beta$? Let $\hat{\beta}(d)$ denote
the MLE of $\beta$ obtained using $d$. A basic requirement is the
following: If $\hat{\beta}(d)$ exists, then $\hat{\beta}(z)$, should
also exist. As we mentioned, the existence of MLE is guaranteed only if
$z$ lies in the interior of convex hull of $\mathit{DS}_n$. As discussed
earlier, even if $d$ lies in the interior of convex hull of $\mathit{DS}_n$, $z$
need not. Thus, directly inputting $z$ into a procedure that estimates
the MLE may lead to meaningless results as the MLE may not exist. See
also, Figure~\ref{fig:MLE} in Section~\ref{compareWithHay} for an
empirical demonstration of nonexistence of MLE when using $z$ to
estimate the parameters.

In the next section, we will see that these issues can be resolved by
modeling the privacy mechanism as a measurement error process, and
computing an estimate $\hat{d}$ of $d$, from the noisy sequence $z$,
that satisfies the conditions in Corollary~\ref
{cor:sufficientconditions} with very high probability. Thus, one of the
advantages of using $\hat{d}$ (instead of $z$) for estimation ensures
that the MLE of $\beta$ exists; see Theorem~\ref{thm:consistent} for a
precise statement. In fact, when using $\hat{d}$ for estimation, not
only does the MLE exist, but the MLE is consistent and asymptotically
normally distributed, as proved in Section~\ref{synthetic}.

\subsection{Maximum likelihood estimation of degree sequence}
\label{mle.ds}

We model the privacy mechanism from Algorithm \ref{alg:lap.degree} as a
measurement error on the degree sequence, and use maximum likelihood
estimation to ``de-noise'' the noisy sequence $z$. The noise addition
process here is regarded as special type of measurement error since we
know the exact distribution of the error.
Hence, despite of the fact that we observe a single sample from the
measurement error process (the degree sequence is released only once),
we can recover an estimate of the original sequence. This takes the
privacy mechanism into account in a principled manner and leads to an
estimate of $d$ that can then be used for inference.
More formally, the
output of Algorithm \ref{alg:lap.degree} generates $n$ random variables
$z_i$, such that $z_i = d_i + e_i$ where $e_i \sim\operatorname
{DLap}(\alpha
)$, for
$i = 1$ to $n$ and $d = \{d_1, \ldots, d_n\} \in \mathit{DS}_n$. Note that
$\alpha$
is known and we treat $d$ as the fixed unknown parameter in $\mathit{DS}_n$.
We propose Algorithm \ref{alg:main} that produces the maximum
likelihood estimator $\hat{d}$ of $d$ from the vector of
noisy degrees $z$, and Theorem~\ref{thm:opt} asserts its correctness.
The proof of Theorem~\ref{thm:opt} is deferred until Section~IV of the supplementary material [\citet{karwaslavsupp}].

\begin{algorithm}[t]
\caption{\textit{Input}: A sequence of integers $z$ of length $n$.
\textit{Output}: A graph $G$ on $n$ vertices with degree sequence $\hat{d}$}\label{alg:main}
\begin{algorithmic}[1]
\STATE Let $G$ be the empty graph on $n$ vertices
\STATE Let $S = \{1, \ldots, n\}$
\WHILE{$|S| > 0$}
\STATE$S = S\setminus T$ where $T = \{i: z_i \leq0\}$
\STATE Let $\operatorname{pos} = |S|$ 
\STATE Let $z_{i^*} = \max_{i \in S}z_i$. Let
$i^* = \operatorname{min} \{i \in S:z_i = z_{i^*}\}$ and let $h_{i^*}=\operatorname{min}(z_{i^*},\operatorname{pos}-1)$
\STATE
Let $\mathcal{I}=$ indices of $h_{i^*}$ highest values in
$z(S\setminus\{i^*\})$ where $z(S)$ is the sequence $z$ restricted to
the index set $S$ 
\STATE Add edge $(i^*,k)$ to $G$ for all $k \in\mathcal{I}$
\STATE Let $z_i=z_i-1$ for all $i \in\mathcal{I}$ and $S = S
\setminus\{i^*\}$
\ENDWHILE
\RETURN$G$
\end{algorithmic}
\end{algorithm}

%
\begin{theorem}[(MLE of degree sequence)]
\label{thm:opt}
Let $z = \{z_i\}$ be a sequence of integers of length $n$ obtained
from Algorithm \ref{alg:lap.degree}. The degree sequence of graph $G$
produced by Algorithm \ref{alg:main} is a maximum likelihood estimator
of $d$.
\end{theorem}

Here, we make some remarks on the complexity of this key result. 
Note that the measurement error model and the corresponding maximum
likelihood estimation of the degree sequence is nonstandard---the
number of parameters to be estimated $(d_i, i = 1, \ldots, n)$ is equal
to the number of observations $(z_i, i = 1, \ldots, n)$, and the
parameter space is discrete and very large---the convex hull of the
parameter set is full dimensional for $n \geq4$.
Computing an MLE of $d$ in the measurement error model is equivalent to
finding a $L_1$ ``projection'' of $z$ on $\mathit{DS}_n$, that is, finding a
graphical degree sequence in $\mathit{DS}_n$ closest to $z$ in terms of the
$L_1$ distance:
%
%
\begin{equation}
\label{eq:opt} \hat{d}= \mathop{\operatorname{argmin}}_{h \in \mathit{DS}_n} \Vert h-z\Vert
_1.
\end{equation}
Here, the parameter set $\mathit{DS}_n$ is a collection of points, and it admits
several characterizations. We found the Havel--Hakimi characterization
to be the most useful in producing an efficient procedure for
estimating the MLE, as evident in the proof of Theorem~\ref{thm:opt};
see Section~IV of \citet{karwaslavsupp}.
In fact, a careful analysis of Algorithm \ref{alg:main} shows that it
is a modified Havel--Hakimi procedure applied to the noisy sequence $z$.

The Havel--Hakimi algorithm is a ``certifying'' algorithm in that it
produces a certificate that a degree sequence is graphical, that is, if
the input to the algorithm is a (graphical) degree sequence, it outputs
a graph that realizes it. Remarkably, our proof of Theorem~\ref
{thm:opt} shows that we can convert such a certifying algorithm into an
algorithm (e.g., Algorithm \ref{alg:main}) that performs $L_1$
``projection'' on the set $\mathit{DS}_n$. We conjecture that our proof
techniques apply to more general polytopes such as the polytope of
degree sequences of bipartite graphs or directed graphs. In cases where
a certifying algorithm like the Havel--Hakimi is available for these
polytopes, our proof techniques can be used to devise algorithms for
$L_1$ optimization over the corresponding set of graphical degree sequences.


Even though the maximum likelihood estimation is equivalent to an $L_1$
projection, there are many differences from the traditional
projection. The set $\mathit{DS}_n$ has ``holes'' in it and is not a convex set.
As an example, every
point whose $L_1$ norm is not divisible by 2 is not included in the
set. Due to this, the $L_1$ projection need not be on the boundary of
the convex hull of $\mathit{DS}_n$. Moreover, there can be more than one degree
sequence that attains the optimal $L_1$ distance. Thus, the MLE of $d$
is actually a set and Algorithm \ref{alg:main} finds a point in this
set. Specifically, the following is true.

%
\begin{lemma}
\label{prop}
Let $d^*$ be the output of Algorithm \ref{alg:main}. Let $Z = \{i:
d^*_i = 0 \mbox{ and }\break z_i < 0 \}$ and $P=\{i: d_i < z_i \mbox{ and }
d_i > 0 \}$, and let $|P| \neq0$. Let $k \in Z$. Then there exists a
degree sequence $d$ such that $d_k > 0$ and $\Vert d^*-z\Vert_1 =
\Vert d-z\Vert_1$.
\end{lemma}
%

Lemma~\ref{prop} [proof
of which is in Section~III of \citet{karwaslavsupp}] shows that the
de-noised degree sequence is not unique. Hence, the noise addition
process provides privacy as the original degree cannot be recovered
exactly. Another way to interpret this result is that the Laplace noise
adds more noise than what is needed to ensure differential privacy, and
Algorithm \ref{alg:main} ``removes'' this additional noise, since
applying Algorithm \ref{alg:main} does not degrade privacy, but
crucially improves utility.

Note that Algorithm \ref{alg:main} is efficient and it runs in time
$O(n\log{n} + m)$ where $n$ is the number of nodes and $m$ is the
number of edges. Algorithm \ref{alg:main} returns a graph $G$ whose
degree sequence is $\hat{d}$, thus, by definition, $\hat{d}$ is
graphical. By randomizing $G$, for example, by using the techniques in
\citet{blitzstein2011sequential} or \citet{ogawa2011graver}, the output
from Algorithm \ref{alg:main} can also be used to generate synthetic
graphs from the uniform distribution of graphs
with a fixed degree sequence, $\mathcal{U}(d)$.

In some cases, especially when some of the $z_i$'s are negative, $G$
may be a disconnected graph.
In such cases, whenever the conditions of Lemma~\ref{prop} are
satisfied, we use it to modify the optimal degree sequence so that it
corresponds to a connected graph. (Note that being the degree sequence
of a connected graph does not ensure that the MLE exists, but the
opposite is true---the MLE of $\beta$ does not exist if the degree
sequence is realized by a disconnected graph.) The proof of Lemma~\ref
{prop} in Section~III of the supplementary material gives
the steps
for the construction of the modified sequence. It is easy to see that
verification of the conditions of Lemma~\ref{prop} and the construction
of the modified sequence takes $O(n \log n)$ time. Hence,
asymptotically, this step does not increase the computational
complexity of Algorithm \ref{alg:main}.
We now proceed to the task of estimating $\beta$ using $\hat{d}$.

\subsection{Asymptotic properties of the private estimate of \texorpdfstring{$\beta$}{beta}}
\label{synthetic}
Let $\hat{d}$ denote the $\varepsilon$-differentially private
estimate of
$d$ obtained by using Algorithms \ref{alg:lap.degree} and \ref
{alg:main}. A private MLE of $\beta$ can be obtained by plugging $\hat
{d}$ in the maximum likelihood equations (\ref{mle.eq}) and solving for
$\beta$; let us denote this estimate by $\hat{\beta}(\hat{d})$. Since
$\hat{d}$ is $\varepsilon$-differentially private, by Lemma~\ref
{thm:post-processing}, $\hat{\beta}(\hat{d})$ is also $\varepsilon
$-differentially private. 
But how does $\hat{\beta}(\hat{d})$ compares to the estimate $\hat
{\beta
}(d)$ obtained from the original degree sequence $d$? We demonstrate
the utility of the proposed private estimate of $\beta$ by proving two
key results in Theorems \ref{thm:consistent} and \ref{thm:clt}, that
is, $\hat{\beta}(\hat{d})$ is consistent and asymptotically normal.

\textit{Consistency}---Consistency of the maximum likelihood estimator
of $\beta$ in the
nonprivate case was shown by \citet{chatterjee2011random}. Here, we
show that our proposed private estimator of $\beta$ is also
consistent, that is one can consistently estimate the parameters of
the $\beta$-model using $\hat{d}$ (as opposed to using $d$).

Theorem~\ref{thm:consistent} shows that using $\hat{d}$ to estimate the
MLE guarantees both the existence of MLE and the uniform consistency
(in contrast to naively using the differentially private output $z$
that does not even guarantee that the MLE exists as discussed in
Sections \ref{degree.mle} and \ref{mle.ds}).

%
\begin{theorem}[(Asymptotic consistency)]
\label{thm:consistent}
Let $G$ be a random graph from the $\beta$-model and let $d = (d_1,
\ldots, d_n)$ be its degree sequence. Let $L = \max_i|\beta_i|$. Let
$\hat{d} = (\hat{d}_{1}, \ldots, \hat{d}_{n})$ be the differentially
private maximum likelihood estimate of $d$ obtained from output of the
Algorithm \ref{alg:main}, and let
\[
\hat{d}_{i} = \sum_{j \neq i}\frac{e^{\hat{\beta}_i + \hat{\beta}_j
}}{1+ e^{\hat{\beta}_i + \hat{\beta}_j }}
\]
be the maximum likelihood equations. Let $C(L)$ be a constant that
depends only on $L$. Then for $\varepsilon_n = \Omega(\frac{1}{\sqrt
{\log n}} ) $, there exists a unique solution $\hat{\beta}(\hat
{d})$ to the maximum likelihood equation such that
\[
\mathbb{P}\biggl( \max_i\bigl|\hat{\beta_i}(
\hat{d}) - \beta_i \bigr| \leq C(L)\sqrt{\frac{\log {n}}{n}} \biggr) \geq1
- C(L)n^{-2}.
\]
\end{theorem}

The proof of Theorem~\ref{thm:consistent} is given in Section~V of the supplementary material [\citet
{karwaslavsupp}]. This
key result implies that asymptotically there is no cost to privacy in
this setting in relation to obtaining valid inference. In particular,
the result shows that for large $n$ and $\varepsilon= \Omega(\frac
{1}{\sqrt{\log n}} )$, the MLE of $\beta$ obtained from $\hat{d}$
exists and is unique and can be estimated with uniform accuracy in all
coordinates. 
In practice, the dependence of $\varepsilon$ on $n$ can be improved by
numerically computing and checking if the tail bound in Lemma~C in the supplementary material [\citet
{karwaslavsupp}], needed for
the proof of Theorem~\ref{thm:consistent}, is satisfied. Thus, this
theorem gives practical guidelines on whether for a given $\varepsilon$
and $n$ combination, the consistency result holds.

Finally, we want to point that if one is allowed to release $d$ many
times using Algorithm \ref{alg:lap.degree}, one can average out the
noise due to the Laplace mechanism and get consistency trivially by
using the law of large numbers. This is not allowed, as the privacy
loss of each release is additive in terms of $\varepsilon$ and would
defeat the purpose of privacy. Hence, to provide meaningful privacy,
the sample size of the private degree sequence is 1, that is, $d$ is
released only once using the Laplace mechanism. Theorem~\ref
{thm:consistent} shows that consistency can still be obtained using a
single private sample of the degree sequence.

\emph{Asymptotic normality}---A central limit theorem for $\hat{\beta
}(d)$ was derived in \citet{Yan01062013}; see also \citet
{Yangeneral}. In
Theorem~\ref{thm:clt}, we derive a similar central limit result for
$\hat{\beta}(\hat{d})$. This distribution can be used to derive
differentially private approximate confidence intervals and perform
hypothesis tests on the parameter estimates. The proof is given in
Section~VI of the supplementary material [\citet{karwaslavsupp}].

Let the covariance matrix of $d = \{d_1, \ldots, d_n\}$ be $V_n = \{
v_{ij}\}$ where
\[
v_{ij} = \frac{\exp{\beta_i + \beta_j}}{(1+\exp{\beta_i + \beta_j})^2}
\]
and
\[
v_{ii} = \sum_{j \neq i, j=1}^n{v_{ij}}.
\]

%
\begin{theorem}[(Asymptotic normality)]
\label{thm:clt}
Let $L = \max_i|\beta_i|$ be a fixed constant and $\varepsilon=
\Omega
(\frac{1}{\sqrt{\log n}} )$. Let $\hat{d}$ be a
differentially private maximum likelihood estimate of $d$ obtained from
Algorithm \ref{alg:main}. Let $\hat{\beta}(\hat{d})$ be the MLE of the
$\beta$-model obtained using $\hat{d}$. For any fixed $r \geq1$, the
random vector
\[
\bigl(\sqrt{v_{11}} \bigl(\hat{\beta}(\hat{d})_1 -
\beta_1 \bigr), \ldots, \sqrt{v_{rr}} \bigl(\hat{\beta}(
\hat{d})_r - \beta_r \bigr) \bigr)
\]
converges to a standard multivariate normal distribution.
\end{theorem}

\section{Releasing graphical degree partitions}
\label{compareWithHay}
In this section, we extend Algorithm \ref{alg:main} to release degree
partitions and compare it with previous work due to \citet{hay2009}.

One can release the degree partition $\bar{d}$ instead of the degree sequence
$d$ in cases where the ordering of the nodes is not important, or one is
interested in the degree distribution (histogram of degrees). The
latter was
the motivation of \citet{hay2009} who instead of releasing the degree
distribution, release the degree partition $\bar{d}$ which has the same
global sensitivity as $d$; thus, Algorithm \ref{alg:lap.degree} can be
used to release a noisy degree partition. Let $z$ be the noisy answer,
that is, $z = \bar{d} +e$. \citet{hay2009} project $z$ onto the set of
integer partitions (nonincreasing integer sequences), which is a
special case of isotonic regression (henceforth referred to as
``Isotone''). They show that this reduces the $L_2$ error. Note,
however, that the output need not be a graphical degree partition, that
is, there may not exist any simple graph corresponding to the output.

To solve this issue, we propose using the following two step algorithm
(referred to as ``Isotone--Havel--Hakimi'' or ``Isotone--HH'') to release
a graphical degree partition.
\begin{longlist}[1.]
\item[1.] Let $\bar{z}$ be the closest integer partition to $z$ in terms of
$L_1$ distance.
\item[2.] Let $\hat{\bar{d}}$ be the output of Algorithm \ref{alg:main} on
input $\bar{z}$.
\end{longlist}

Unlike the case of degree sequence, this procedure does not estimate an
MLE of $\bar{d}$. However, Corollary~\ref{thm:opt.part} shows that the
estimate is still optimal in sense of the $L_1$ error, and more
importantly, it is a point in $\mathit{DP}_n$ that is closest to $\bar{z}$. The
proof of Corollary~\ref{thm:opt.part} appears in Section~VII of the supplementary material [\citet
{karwaslavsupp}]. 

%
\begin{corollary}
\label{thm:opt.part}
Let $\bar{z} = \{\bar{z}_i\}$ be a sequence of nonincreasing integers
of length~$n$. The degree partition of graph $G$ output by Algorithm
\ref{alg:main} on input $\bar{z}$ is a solution to the optimization
problem $\operatorname{argmin}_{h \in \mathit{DP}_n} \Vert h-\bar{z}\Vert_1$.
\end{corollary}

Release of synthetic graphs here follows as discussed in
Section~\ref{mle.ds}.

\section{Simulation results}
\label{experiments}
In this section, we evaluate the finite sample properties of the
differentially private estimator of $\beta$. 
We perform two sets of
experiments. In the first set, we compare the utility of
\citet{hay2009} with our algorithm when releasing degree partitions
$\bar
{d}$. In the second set of experiments, we estimate $\beta$ using the
private estimate $\hat{d}$ of Algorithm \ref{alg:main} and compare it
with the estimates obtained by using the nonprivate degree $d$. We use
three networks, two real and one simulated, described below.
\begin{longlist}[1.]
\item[1.] Sampson Monastery Data [\citet{sampson1968novitiate}]---This is a
real network of relationship between monks in a monastery. It consists
of social relations among a set of 18 monks. The original dataset was
asymmetric and collected for three time periods. In this study, we
symmetrize the network by using the upper triangular adjacency matrix
of time period 1. There are $18$ nodes and $35$ edges in this network.

\item[2.] Karate Dataset [\citet{zachary1977}]---This is a real network of
friendships between 34 members of a karate club at a US university in
the 1970. It has $78$ edges and $34$ nodes. 

\item[3.] Likoma Island [\citet{helleringer2007sexual}]---This is a simulated
network of number of sexual partners of people living in the Likoma
island. \citet{helleringer2009likoma} describes the study and data
collection procedures based on a survey. Using the estimated degree
sequence [obtained from the survey data and given in \citet
{helleringer2009likoma}], we simulated a random network with the fixed
degree sequence. The simulated network consists of $250$ nodes and
$248$ edges. 
\end{longlist}

\textit{Releasing $\bar{d}$ to estimate $\beta$}: The goal of these
experiments is to compare \textit{isotone} and \textit{isotone--hh} algorithms
for releasing differentially private verstions of degree partitions
$\bar{d}$. We evaluate these algorithms on two metrics. The first
metric is the probability of the event $R$ where $R = \{ \hat{\beta}(y)
\mbox{ exists}\}$, where $y$ is output of the mechanism. The second
metric is the median $L_1$ error between $\bar{d}$ and $y$ for fixed
$\bar{d}$, that is, $\operatorname{err}(\bar{d}) = \operatorname{median}[|\bar{d}-y|]$.
For each network and a fixed value of privacy parameter $\varepsilon$,
$\bar{d}$
is released $B=500$ times using \textit{isotone} and our
\textit{isotone--hh} procedure. Note that even though each release of
$\bar{d}$ is
$\varepsilon$-edge differentially private, the entire simulation study
is $500
\varepsilon$-edge differentially private. In practice, $\bar{d}$ will be
released only once. However, in the experiments, we are interested in
evaluating the frequentist properties of the procedure, and hence we release
the degree partition multiple times. Using these released degree partitions,
we compute $P(R)$ and $\operatorname{err}(\bar{d})$. This procedure is repeated for
different levels of $\varepsilon$ varying from $0$ to $4$, for all three
datasets. Note that a larger $\varepsilon$ means lower noise and less
privacy. Figure~\ref{fig:MLE} shows a plot of $P(E)$ and $\operatorname{err}(\bar{d})$
normalized by the number of nodes for varying levels of $\varepsilon$.

%
\begin{figure}

\includegraphics{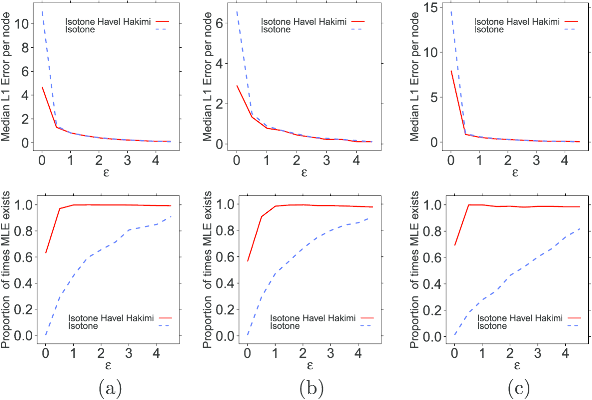}

\caption{Comparison of ``Isotone'' and ``Isotone--HH'' to release $\bar
{d}$. The plots show the $L_1$ error and the probability that the MLE
exists for varying levels of $\varepsilon$ for three different
networks. \textup{(a)}~Karate; \textup{(b)}~Sampson; \textup{(c)} Likoma.}
\label{fig:MLE}
\end{figure}

%

As expected, for both algorithms, as $\varepsilon$ increases, $P(R)$ increases
and the median $L_1$ error decreases. In many cases, the MLE of the
output of
\textit{isotone} fails to exist as it lies outside the convex hull of
$\mathit{DP}_n$. $P(R)$ is significantly higher for \textit{isotone--hh} for all
three datasets. For instance for the Karate dataset, $P(R)$ quickly
approaches $1$ as $\varepsilon$ increases, when using the \textit
{isotone--hh} algorithm, where as it never reaches $1$ when using the
$\mathit{isotone}$ algorithm. The other two datasets exhibit similar behavior.
We can also see that for the Likoma dataset, the gap between the two
algorithms in terms of $P(R)$ is much higher when compared to the other
two datasets. More specifically, when using the \textit{isotone}
algorithm, $P(R)$ increases slowly with $\varepsilon$ for the Likoma
dataset when compared to the other two datasets. On the other hand,
when using the \textit{isotone--hh} algorithm, $P(R)$ increases quickly
with $\varepsilon$ for all three datasets. A possible explanation for the
behavior of the \textit{isotone} algorithm is that the Likoma data are
sparse. Recall that $P(R)$ is $0$ if the noisy sequence lies outside
$\operatorname{Conv}(\mathit{DP}_n)$ (see Theorem~\ref{thm:beta.mle.exist}).
Due to the
sparsity of Likoma data, the degree partition is close to the boundary
of $\operatorname{Conv}(\mathit{DP}_n)$. In this case, adding Laplace noise
puts the degree
partition outside $\operatorname{Conv}(\mathit{DP}_n)$, and the post-processing
step of \textit
{isotone} is not sufficient to get a sequence inside $\operatorname
{Conv}(\mathit{DP}_n)$, and
hence $P(R) = 0$ for such instances.

When considering the median $L_1$ error, the \textit{isotone--hh}
algorithm not only provides an increased probability that the MLE
exists, but also provides more accurate estimates of $\bar{d}$,
especially for smaller levels of $\varepsilon$. For instance, for
$\varepsilon= 0.1$, for the Karate dataset, the median $L_1$ error per
node in estimating the degree is $4$ for the \textit{isotone--hh}
whereas it is greater than $10$ for the \textit{isotone} algorithm.
Thus, we can see that \textit{isotone--hh} offers more ``utility'' in
terms of both estimating the MLE, and also in terms of the $L_1$ error.

%
\begin{figure}

\includegraphics{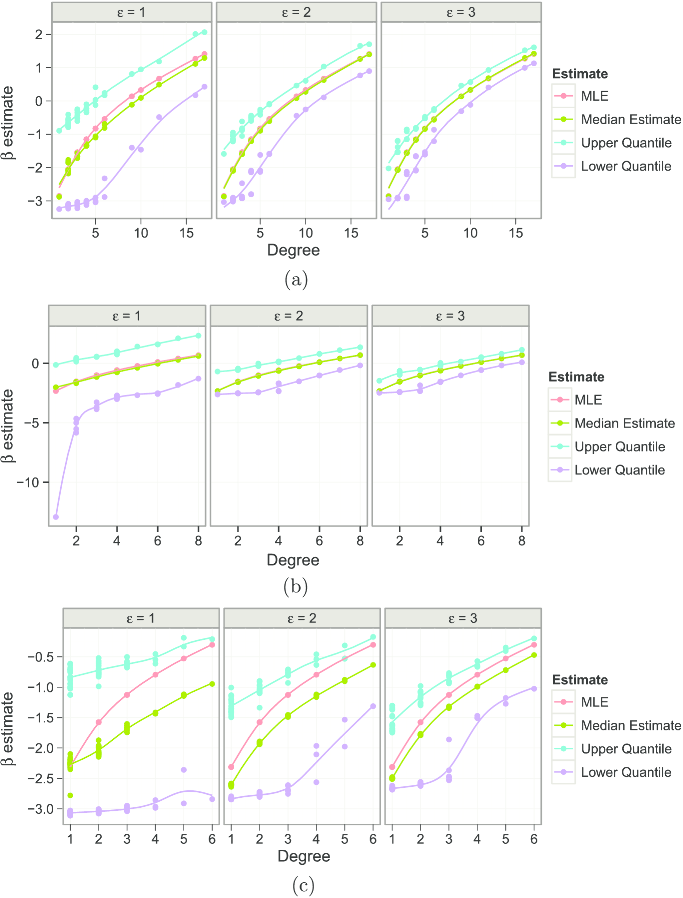}

\caption{Comparison of differentially private estimate of $\beta$ with
the MLE for three different datasets. The plots show the median and the
upper (95{th}) and the lower (2.5{th}) quantiles. \textup{(a)} Karate
data; \textup{(b)}~Sampson data; \textup{(c)} Likoma island data.}
\label{fig:betadegree}
\end{figure}

\textit{Estimation of $\beta$ using $d$}: In the second set of experiments,
we evaluate how close $\hat{\beta}(\hat{d})$ is to
$\hat{\beta}(d)$. Here, $\hat{\beta}(d)$ is the estimate of $\beta
$ obtained
by using the original degree sequence and $\hat{\beta}(\hat{d})$
is the estimate of $\beta$ obtained by using the private degree sequence
$\hat{d}$ obtained from the output of Algorithm
\ref{alg:main}. Figure~\ref{fig:betadegree} shows a plot of the
estimates of
$\beta$ on the $y$ axis and degree on the $x$ axis. The red line indicates
$\hat{\beta}(d)$ and the green line indicates the median estimate of
$\hat{\beta}(\hat{d})$. Also plotted are the upper (95th) and
the lower (2.5{th}) quantiles of the estimates. The results show
that the
median estimate of $\beta(\hat{d})$ is very close to $\beta(d)$ and lies
within the $95$ percent quantiles of the estimates. Moreover, as
expected, as
$\varepsilon$ increases, the variance in the estimates get smaller.
The median
private estimates of $\beta$ for the Karate and the Sampson dataset
are very
close to the nonprivate MLE. However, the private estimates of $\beta$ for
the Likoma dataset have higher variance and are farther from MLE of
$\beta(d)$ due to the fact that the Likoma graph is sparse and the
$\beta$-model does not fit the original data very well. This suggests that
the $\beta$-model may not be a \emph{robust} model for sparse networks
in the following sense. If the network is very sparse, the degree
sequence of the original data may lie
close to the boundary of $\operatorname{Conv}(\mathit{DS}_n)$. Due to this,
adding or removing
a small number of edges may cause the degree sequence to end up being
\underline{on} the boundary.

\section{Conclusions and future work}
\label{conclusion}
In this paper, we characterize the conditions for the existence of MLE
of the degree partition and the degree sequence $\beta$-model that lead
to an efficient quadratic time algorithm. Motivated by the privacy
problem of sharing confidential data under rigorous privacy guarantees,
that often falls short of satisfying data utility, we present
techniques to perform valid and differentially private statistical
inference with the $\beta$-model of random graphs and to release
differentially private synthetic graphs from the $\beta$-model. We
present an efficient maximum likelihood algorithm to re-estimate the
original degree sequence from a noisy sequence released by a
differentially private mechanism. We showed that this estimated degree
sequence can be used to obtain a consistent and asymptotically normally
distributed estimates of the parameters of the $\beta$-model, and thus
incur no cost due to privacy from utility perspective. Using the
example of the $\beta$-model, we showed that 
the noisy sufficient statistics $z$ must be post-processed (or
projected) in
an appropriate manner by taking the noise mechanism into account in
order to obtain optimal inference. In particular, by treating the
privacy mechanism as a nonlinear measurement error model, one can
estimate the sufficient statistics from their noisy counterparts and
obtain optimal inference. This also ensures that existing methods for
maximum likelihood estimation do not break.

We would like to note again, in light of Corollary~\ref
{cor:sufficientconditions}, that in general, using noisy sufficient
statistics $z$ of any model instead of the true sufficient statistics
may lead to inconsistent estimates, in particular, nonexistence of
MLE. A key issue is that the noisy statistic $z$ usually lies in
$\mathbb{R}^n$ whereas the validity of many inference procedures (such
as existence of MLE and consistency) is guaranteed only when $z$ lies
in some set $S \subset\mathbb{R}^n$, typically the convex hull of
sufficient statistics of the associated model, for example, $S =
\operatorname{Conv}(\mathit{DS}_n)$.
In some cases, $z$ is post-processed and projected onto a set $S'$; the
choice of $S'$ is motivated with a goal of imposing some reasonable
constraint on the noisy statistic, and to reduce the $L_2$ error
between the noisy and the original statistics. But usually, $S \neq
S'$. We showed with the degree partition example that such approach
does not even guarantee the existence of MLE, let alone consistency.
Thus more carefully designed and provable methods are needed to
guarantee utility, keeping in mind the end goals of statistical
inference (e.g., estimation of parameters, and not just statistics).

We demonstrated that significant gains in utility can be made by using
a two step technique of (a) ``de-noising'' the noisy statistic using
maximum likelihood estimation on the measurement error model and (b)
estimating the MLE of the parameter of interest using the de-noised
version of the statistic. Note that the first step is equivalent to
``projecting'' the noisy statistic onto the lattice points of the
corresponding marginal polytope.
While this two step procedure guarantees that the MLE of the parameter
exists, a priori, these is no reason to believe that the estimates are
also consistent and asymptotically normal. But we prove, remarkably, in
the case of $\beta$-model, that they are. We believe that this
principled two step approach could be applicable in other settings, and
would lead to not only existence of MLE but also consistency and
asymptotic normality. An interesting class of models to extend these
techniques to are the general class of discrete exponential families
and in particular, various families of $\beta$-models such as the Rasch
models of bipartite graphs [e.g., \citet{rinaldo2011maximum}], models
based on weighted degree sequences such as those studied in \citet
{hillar2013maximum} and degree sequences of directed graphs, and
finally the class of log-linear models where \citet{FienbergRinYang}
have already demonstrated some of the above mentioned issues with
estimations done in a privacy-preserving manner.

There are several challenges in extending our principles to the above
mentioned class of models. One of the key challenges is, for each of
these families, finding a description of the marginal polytope $S$ that
would allow the ``de-noising'' step; the marginal polytope is a complex
combinatorial object associated with the existence of MLE and is a
focus of many studies; see, for example, \citet{rinaldo2009geometry},
but its characterization is often nontrivial. One avenue for further
work is to use the technique of asymmetrization of a polytope, as done
in this paper, to derive efficient conditions for the existence of MLE
for generalized $\beta$-models. Once such a description is found, the
next challenge is to devise an efficient algorithm for ``projecting''
the noisy statistic onto the set of lattice points of the marginal
polytope. The projection can be informed by the measurement error
model. In our case, the significant contribution is achieved, by
combing these two steps---finding the ``right'' description of $S$ and
a projection algorithm---into one step. We do this by using an
efficient algorithmic description of the lattice points of the marginal
polytopes (e.g., the Havel--Hakimi algorithm [\citet{havel,hakimi}]
provides such a description for degree sequences) and somewhat
surprisingly, converting such a description into an efficient
projection algorithm. Such efficient descriptions do no exist for the
more general class of discrete exponential families [e.g., see \citet
{hillar2013maximum} and \citet{engstrom2010polytopes}] and is an
interesting direction of future work that goes beyond private
estimation and warrants an independent inquiry.

In cases where de-noising is not possible, for example with more
general graph statistics, how can we capture the noise infusion due to
privacy or some other mechanism? An
alternative is to develop new statistical procedures that integrate
the noise addition process into the likelihood by using missing data
techniques, for example, see \citet{karwa2014differentially} for
differentially private estimation of exponential random graph models.
But such solutions may be computationally expensive and currently lack
theoretical properties.

\section{Proofs}
\label{proofs}

\subsection{Proof of Theorem \texorpdfstring{\protect\ref{thm:beta.mle.exist}}{1}}
\label{proof:beta.mle.exist}
The key technique to prove this result is to use the polytope of degree
partitions to characterize the boundary of the polytope of degree
sequence, $\operatorname{Conv}(\mathit{DS}_n)$. We will need the following
result from \citet
{mahadev1995threshold} that characterizes the boundary of
$\operatorname{Conv}(\mathit{DS}_n)$.

%
\begin{lemma}[{[Lemma~3.3.13 in Mahadev and Peled (\citeyear
{mahadev1995threshold})]}]
\label{lem:bdofDSn}
Let $d$ be a degree sequence of a graph $G$ that lies on the boundary
of $\operatorname{Conv}(\mathit{DS}_n)$. Then there exist nonempty and disjoint
subsets $S$
and $T$ of $\{1, \ldots, n\}$ such that:
\begin{longlist}[1.]
\item[1.]$S$ is clique of $G$;
\item[2.]$T$ is a stable set of $G$;
\item[3.] Every vertex in $S$ is adjacent to every vertex in $(S\cup T)^c$
in $G$;
\item[4.] No vertex of $T$ is adjacent to any vertex of $(S\cup T)^c$ in $G$.
\end{longlist}
\end{lemma}

\emph{Part} (i)---\emph{MLE of the degree partition $\beta$-model}: By
Theorem~9.13 in \citet{nielsen1978information}, the MLE $\hat{\beta
}(\bar
{d})$ exists iff $\bar{d} \in \mathit{ri}(\operatorname{Conv}(\mathit{DP}_n))$. Here,
$\mathit{ri}(\operatorname{Conv}(A))$
denotes the relative interior of the convex hull of $A$. To prove the
first part of the theorem, note that the following system of
inequalities along with the constraint $d_1 \leq d_2 \leq\cdots\leq d_n$
describe the faces the convex hull of degree partitions [see
Theorem~1.3 in \citet{bhattacharya2006polytope}]:
\begin{longlist}[1.]
\item[1.]
\[
\bar{d}_i > 0 \quad\mbox{and}\quad \bar{d}_i < n-1\qquad \forall i\quad
\mbox{and},
\]
\item[2.]
\[
\sum_{i=1}^k{\bar{d}_i} -
\sum_{i=n-l+1}^{n}{\bar{d}_i} <
k(n-1-l) \qquad\mbox{for } 1 \leq k+1 \leq n.
\]
\end{longlist}

Thus, the ordering constraints also define $n-1$ faces of the polytope
given by $d_{i+1}- d_{i} \leq0$. For a degree partition to be in the
interior of $\operatorname{Conv}(\mathit{DP}_n)$, it must hold that $\bar{d}_1
> \bar{d}_2 >
\cdots> \bar{d}_n$. This is possible only if each $\bar{d}_i = n-i$.
However, such a sequence is not realizable (and hence not a degree
sequence) as $\bar{d}_n = 0$ and $\bar{d}_1 = n-1$. Hence, there is no
degree partition that lies in the interior of $\operatorname
{Conv}(\mathit{DP}_n)$, and the MLE
for the degree partition $\beta$-model never exists when we observe
only one graph.

\emph{Part} (ii): We have to show that if the MLE $\hat{\beta}(d)$
exists, then $\bar{d}$ satisfies the system (\ref{mle.ineq}). Recall that
the MLE $\hat{\beta}(d)$ exists iff $d \in \mathit{ri}(\operatorname
{Conv}(\mathit{DS}_n))$. Also, note
that $d \in \mathit{ri}(\operatorname{Conv}(\mathit{DS}_n))$ iff
%
%
\begin{equation}\qquad
\label{eq:desDS} \sum_{i \in S }{d_i} - \sum
_{i \in T}{d_i} < |S|\bigl(n-1-|T|\bigr)\qquad \forall S, T
\subset[n], S \cup T \neq\varnothing, S \cap T = \varnothing.
\end{equation}
For example, see Theorem~3.3.17 in \citet{mahadev1995threshold}.


We show that the system of inequalities in (\ref{eq:desDS}) are
permutation invariant, that is, if $d$ satisfies (\ref{eq:desDS}), then
$\pi d$ also satisfies (\ref{eq:desDS}), where $\pi$ is any permutation
on $[n] = \{1, \ldots, n\}$. To see this, let $(\mathcal{S},\mathcal
{T}) = \{(S,T)\}$ be the set of all possible sets $S$ and $T$ such that
$S, T \subset[n] = \{1, \ldots,n\}$, $S \cup T \neq\varnothing$, $S
\cap T = \varnothing$. First, note that if $(S,T) \in(\mathcal
{S},\mathcal{T})$, then $(T,S) \in(\mathcal{S},\mathcal{T})$. Also,
note that $(\mathcal{S},\mathcal{T})$ is closed under permutations,
that is, if $(S,T) \in(\mathcal{S},\mathcal{T})$, and if $\pi$ is any
permutation on $[n]$, then $ (\pi S, \pi T) \in(\mathcal{S},\mathcal{T})$.


Now assume that $d \in \mathit{ri}(\operatorname{Conv}(\mathit{DS}_n))$, we need to show
that $\bar{d}$
satisfies the system of inequalities (\ref{mle.ineq}). Note that $d$
satisfies (\ref{eq:desDS}). By the fact that these inequalities are
permutation invariant, any permutation of $d$ also satisfies (\ref
{eq:desDS}). Hence, as $\bar{d} = \pi d$ for some permutation $\pi$,
(\ref
{eq:desDS}) is true for $\bar{d}$.

Taking $S = \{1, \ldots, k\}$ and $T = \{n-l+1, \ldots, n\}$ gives the
second set of inequalities in (\ref{mle.ineq}). Taking $S= \{i\}, T=
\varnothing$ gives $\bar{d}_i < n-1$ and taking $S = \varnothing, T
= \{i\}
$ gives $\bar{d}_i > 0$. 

\emph{Part} (iii): Assume that $\bar{d}$ satisfies the system (\ref
{mle.ineq}). We will show that $\bar{d}$ does not lie on the boundary of
$\operatorname{Conv}(\mathit{DS}_n)$. This will imply that $\bar{d} \in
\mathit{ri}(\operatorname{Conv}(\mathit{DS}_n))$, which
implies that $\bar{d}$ satisfies the inequalities (\ref{eq:desDS}). By
the permutation invariance of the system (\ref{eq:desDS}), $\pi\bar{d} =
d$ also satisfies (\ref{eq:desDS}), from which the result follows.

All that is remaining to be shown is that $\bar{d}$ does not lie on the
boundary of $\operatorname{Conv}(\mathit{DS}_n)$. The boundary of
$\operatorname{Conv}(\mathit{DS}_n)$ is characterized
by Lemma~\ref{lem:bdofDSn}. Let $G$ be a graph that realizes $\bar{d}$,
hence $G$ is such that there exist disjoint subsets of $\{1, \ldots,
n\}
$ $S$ and $T$ satisfying conditions of Lemma~\ref{lem:bdofDSn}.

Let $i \in S$, then $\bar{d}_i \geq(|S|-1) + |(S\cup T)^c| = n -
|T|-1$ (by conditions 1 and 3 of Lemma~\ref{lem:bdofDSn}). Let $i \in
T$ then $\bar{d}_i \leq|S|$. Finally if $i \in(S \cup T)^c$, then
$\bar{d}_i \geq|S|$ (by condition 3 in Lemma~\ref{lem:bdofDSn}) and
$\bar{d}_i \leq|S| + |(S\cup T)^c|-1 = n - |T| - 1$ (by condition 4 in
Lemma~\ref{lem:bdofDSn}). Putting these together, we get the following:
%
%
\begin{eqnarray}
\label{eq:char} %
0 &\leq&\bar{d}_i \leq|S|,\qquad i \in T,
\nonumber
\\
|S| & \leq&\bar{d}_i \leq n - |T| -1,\qquad i \in(S\cup T)^c,
\\
n - |T|-1 &\leq&\bar{d}_i \leq n-1,\qquad i \in S.
\nonumber
\end{eqnarray}

Now note that $\bar{d}_1 \leq\bar{d}_2 \leq\cdots\leq\bar{d}_n$. Hence,
the only possible choice for $S$ and $T$ are $S = \{1, \ldots, k \}$
and $T = \{n-l+1, \ldots, n\}$ where $k = |S|$, $l = |T|$, $1 \leq k +l
\leq n$. No other combinations of $S$ and $T$ exist, due to the
characterization of $\bar{d}$ given in equation (\ref{eq:char}).
Next, since $\bar{d}$ is on the boundary of $\operatorname
{Conv}(\mathit{DS}_n)$, it holds that
$\sum_{i \in S }{d_i} - \sum_{i \in T}{d_i} = |S|(n-1-|T|)$ for all
such $S$ and $T$ described above. However, we are given that this is
not true. Hence $\bar{d}$ must lie in the interior of $\operatorname
{Conv}(\mathit{DS}_n)$.

\section*{Acknowledgments}
Karwa was a graduate student at the Department of Statistics,
Pennsylvania State University when the paper was initially submitted.
The authors would also like to thank anonymous referees, editor,
Alessandro Rinaldo and Johannes Rau for helpful feedback.

\begin{supplement}[id=suppA]
\stitle{Supplement to ``Inference using noisy degrees: Differentially
Private $\beta$-model and synthetic graphs''}
\slink[doi]{10.1214/15-AOS1358SUPP} 
\sdatatype{.pdf}
\sfilename{aos1358\_supp.pdf}
\sdescription{This supplementary material contains the proof of the key Theorems \ref{thm:opt}, \ref{thm:consistent} and \ref{thm:clt} from the paper.}
\end{supplement}

%


%


\printaddresses
\end{document}